\documentclass[conference,letterpaper]{IEEEtran}

\usepackage[utf8]{inputenc} 
\usepackage[T1]{fontenc}
\usepackage{url}
\usepackage{ifthen}
\usepackage{cite}
\usepackage[cmex10]{amsmath} % Use the [cmex10] option to ensure complicance
\usepackage{comment}
\usepackage{graphicx} 
\usepackage{amssymb}
\usepackage{amsthm}
\usepackage{hyperref}
\usepackage{algorithm}
\usepackage{algpseudocode}
\usepackage{xcolor}
\usepackage{bbm}

\newtheorem{theorem}{Theorem}
\theoremstyle{definition}
\newtheorem{remark}{Remark}
\newtheorem{lemma}{Lemma}

\newtheorem{definition}{Definition}

\newcommand{\pare}{\eta}

\newcommand{\gdot}{g(.)}

\newcommand{\acce}{\mathcal{A}_{\pare}}

\newcommand{\bu}{\mathbf{u}}
\newcommand{\by}{\mathbf{y}}
\newcommand{\bn}{\mathbf{n}}

\newcommand{\mA}{\mathcal{A}}
\newcommand{\est}{\mathsf{est}}
\newcommand{\DC}{\mathsf{DC}}

\newcommand{\MSE}{\mathsf{MSE}}
\newcommand{\PA}{\mathsf{PA}}

\newcommand{\mB}{\mathcal{B}}

\newcommand{\bestu}{\mathsf{U}^*_{a,b}}

\newcommand{\bestetaquant}{\pare^*_{q}}
\newcommand{\bestetaalg}{\hat{\eta}^*_{\delta,\lambda}}
\newcommand{\bestetaalgnew}{\hat{e}^*_{\delta,\lambda}}

\newcommand{\han}[1]{{\color{blue}#1}}
\begin{document}
\title{Game of Coding With an Unknown Adversary} 

% %%% Single author, or several authors with same affiliation:
% \author{%
%  \IEEEauthorblockN{Andrew R.~Barron}
%  \IEEEauthorblockA{Department of Statistics and Data Science\\
%                    Yale University\\
%                    New Haven, CT, USA\\
%                    Email: andrew.barron@yale.edu}
% }

%%% Several authors with up to three affiliations:
\author{%
  \IEEEauthorblockN{Hanzaleh Akbarinodehi}
  \IEEEauthorblockA{
                    University of Minnesota, Twin Cities\\
                    Minneapolis, MN, USA}
                    \and
  
  \IEEEauthorblockN{Parsa Moradi}
  \IEEEauthorblockA{
                    University of Minnesota, Twin Cities\\
                    Minneapolis, MN, USA}
  \and
  \IEEEauthorblockN{Mohammad Ali Maddah-Ali}
  \IEEEauthorblockA{
                    University of Minnesota, Twin Cities\\
                    Minneapolis, MN, USA}
}

\maketitle

\begin{abstract}
Motivated by emerging decentralized applications, the \emph{game of coding} framework has been recently introduced to address scenarios where the adversary's control over coded symbols surpasses the fundamental limits of traditional coding theory. 
Still, the reward mechanism available in decentralized systems, motivates the adversary to act rationally. While the decoder, as the data collector (DC), has an acceptance/rejection mechanism, followed by an estimation module, the adversary aims to maximize its utility, as an increasing function of (1) the chance of acceptance (to increase the reward), and (2) estimation error. On the other hand, the decoder also adjusts its acceptance rule to  maximize its own utility, as (1) an increasing function of the chance of acceptance (to keep the system functional), (2) decreasing function of the estimation error. Prior works within this framework rely on the assumption that the game is complete—that is, both the DC and the adversary are fully aware of each other’s utility functions. However, in practice, the decoder is often unaware of the utility of the adversary. 

To address this limitation, we develop an algorithm enabling the DC to commit to a strategy that achieves within the vicinity of the equilibrium,  without knowledge of the adversary's utility function. Our approach builds on an observation that at the equilibrium, the relationship between the probability of acceptance and the mean squared error (MSE) follows a predetermined curve independent of the specific utility functions of the players. By exploiting this invariant relationship, the DC can iteratively refine its strategy based on observable parameters, converging to a near-optimal solution. We provide theoretical guarantees on sample complexity and accuracy of the proposed scheme.

\end{abstract}

\section{Introduction}

Coding theory is vital for ensuring data reliability and integrity in communication, computing, and storage. However, a fundamental challenge is its reliance on trust assumptions for ensuring successful data recovery. In such scenarios, where computations or data are distributed across \( N \) nodes, some nodes, denoted by \( \mathcal{H} \), are honest and adhere to the protocol, while others, represented by \( \mathcal{T} \), are adversarial and deviate arbitrarily.
For error-free recovery, $|\mathcal{H}|$ must exceed $|\mathcal{T}|$. For example, in repetition coding, reliable recovery requires \( |\mathcal{H}| \geq |\mathcal{T}| + 1 \). For Reed-Solomon codes~\cite{SudanBook}, which encode \( K \) symbols into \( N \) coded symbols, recovery is guaranteed if \( |\mathcal{H}| \geq |\mathcal{T}| + K \). Similarly, in Lagrange codes~\cite{yu2019lagrange}, for a polynomial  of degree \( d \), recovery requires \( |\mathcal{H}| > |\mathcal{T}| + (K-1)d \).  Similar constraints apply in analog scenarios  \cite{roth2020analog, ZamirCoded,  jahani2018codedsketch,BACC}. 

However, in emerging applications such as decentralized (blockchain-based) machine learning (DeML),  the traditional trust assumptions of coding theory are often impractical.  In these contexts, due to the lack of centralized control and the inherently adversarial nature of the environment, trust becomes a scarce resource~\cite{sliwinski2019blockchains, han2021fact, gans2023zero}. Nonetheless, decentralized platforms introduce an incentive-driven dynamics that offer new opportunities to overcome these limitations. Specifically, these systems reward contributions as long as they remain functional (\emph{live}). This incentivizes adversaries to ensure that data is recoverable by the decoder, as the data collector (DC), rather than engaging in complete denial-of-service attacks. As a result, adversaries are driven to keep the system live while introducing estimation errors to maximize their own benefits. 

Motivated by this observation and focusing on applications such as DeML and oracles, the game of coding framework was introduced in \cite{nodehi2024game}. In this framework, the interaction between the DC and the adversary is modeled as a game, where both players aim to maximize their respective utility. These utility functions are based on two key factors: (1) the probability of the DC accepting the inputs and (2) the error incurred in estimating the original data when inputs are accepted. Initially, \cite{nodehi2024game} examines a two-node system, depicted in Fig. \ref{fig:Two_node_model}, consisting of one honest node, one adversarial node, and a DC. The DC aims to estimate a random variable \(\mathbf{u} \). However, the DC does not have direct access to \(\mathbf{u}\) and relies on the two nodes for an estimation.
The honest node sends \(\mathbf{y}_h = \mathbf{u} + \mathbf{n}_h\) to the DC, where \(- \Delta \leq \mathbf{n}_h \leq \Delta\) for some \(\Delta \in \mathbb{R}\). The adversary  sends \(\mathbf{y}_a = \mathbf{u} + \mathbf{n}_a\), where \(\mathbf{n}_a\) is an arbitrary noise independent of \(\mathbf{u}\). Upon receiving \((\mathbf{y}_1, \mathbf{y}_2)\), the DC decides whether to accept or reject the inputs. While the acceptance rule \(|\mathbf{y}_1 - \mathbf{y}_2| \leq 2\Delta\) seems reasonable, it may not align with the DC's objectives. Specifically, this narrow acceptance region might reject inputs that could provide an accurate enough estimate of \(\mathbf{u}\). To optimize its utility, the DC can expand the acceptance region to \(|\mathbf{y}_1 - \mathbf{y}_2| \leq \eta \Delta\), for some \(\eta \geq 2\). 
In \cite{nodehi2024game}, the equilibrium of the game under the Stackelberg model \cite{simaan1973stackelberg} is analyzed, where the DC acts as the leader, committing to a strategy \(\eta\), and the adversary as the follower responds with \(\mathbf{n}_a\). The optimal strategies for both are derived in \cite{nodehi2024game}.  Subsequently, \cite{akbari2024game} extends this analysis to scenarios involving $N>2$ nodes. It demonstrated that the adversary’s utility at equilibrium does not increase with more nodes. This property ensures \emph{Sybil resistance}, where adversarial strategies gain no benefit from having more nodes.

However, one of the main challenges in \cite{nodehi2024game} and  \cite{akbari2024game} is the assumption that the game is complete, meaning both the DC and the adversary knows each other’s utility functions. While it is reasonable to assume that the adversary, as a powerful player, is aware of the DC's utility function, the reverse assumption is often unrealistic. In practical scenarios, the DC typically lacks the knowledge of the adversary's utility function. 
This lack of perfect information, known as the {\emph game of incomplete information}, introduces uncertainty for the DC in choosing the optimal strategy.  
The incomplete Stackleberg games, where the follower's utility function is unknown to the leader,  have been extensively explored in the literature~\cite{conitzer2006computing,gan2023robust, letchford2009learning, blum2014learning, peng2019learning, sessa2020learning,balcan2015commitment,  haghtalab2022learning, dong2018strategic, kleinberg2003value}. The solutions can be divided into two categories:   In the first category,  no attempt is made to reduce this uncertainty, and the leader's strategy is chosen based on optimizing the expected payoff (in a Bayesian formulation~\cite{conitzer2006computing}), or the worst payoff (in a robust setting~\cite{gan2023robust}). In the second category,  the players engage in a sequential game, where the leader updates its belief about the follower's payoff as the game progresses~\cite{letchford2009learning, blum2014learning, peng2019learning, sessa2020learning,balcan2015commitment,  haghtalab2022learning, dong2018strategic, kleinberg2003value}.  However, due to the inherent complexity of the game of coding, even for the most straightforward scenarios, these approaches pose significant challenges. 

In this paper, we develop an algorithm that enables the DC to commit to a strategy that is sufficiently effective, even without knowledge of the adversary's utility function. At first glance, this may seem like an impossible task, as determining the optimal strategy typically requires knowledge of the adversary's best response to each strategy the DC might commit to. Without this information, it appears infeasible to identify the optimum strategy.  Still there is one approach. The DC can commit to a sequence of strategies, one at a time, for which naturally the adversary will react to it by choosing its best corresponding strategy. Then, the DC can estimate the probability of acceptance and the mean squared error (MSE) for each case, and compute its own utility function for each strategy, and over time converge to the best strategy. However, the challenge lies in the fact that estimating the MSE requires the DC to know the value of  \( \mathbf{u} \),  which is not available. 

To resolve this, we utilize an intriguing observation that for each strategy committed by the DC and the corresponding response chosen by the adversary, the relationship between the probability of acceptance and the MSE, as determined by the adversary’s response, follows a specific and known relation. Remarkably, this curve is \emph{independent} of the specific utility functions of both the DC and the adversary. By exploiting this invariant relationship, the DC can estimate its own utility based on observable parameters, enabling it to iteratively refine its strategy. This approach allows the DC to converge to a strategy that is close to the optimal solution in the absence of complete information about the adversary’s utility function.
This paper takes the first step toward addressing this challenge by presenting algorithms to resolve the issue, proving their correctness, and evaluating their sample complexity.

The rest of the paper is organized as follows. In Section \ref{Formal Problem Setting}, we present the formal problem formulation. In Section \ref{sce:Preliminaries}, we discuss the preliminaries and key insights from previous works. Section \ref{Main results} outlines the main results and provides clarifying examples to illustrate their implications. Finally, in Section \ref{Proofs of main theorems}, we present the proofs of the main results.

\textbf{Notation}: For random variables we use boldface letters.  
   For any set $\mathcal{S}$ and any function $f(.)$, the output of $\underset{x \in \mathcal{S}}{\arg\max} ~f (x)$ and $\underset{x \in \mathcal{S}}{\arg\min} ~f (x)$ is a set.
For any \( x \in \mathbb{R} \), \( \lceil x \rceil \) denotes the smallest integer greater than or equal to \( x \). Similarly, for any \( n \in \mathbb{N} \), \( [n] \) denotes the set of all integers from 1 to \( n \).

\section{Problem formulation}\label{Formal Problem Setting}
In this section we present the problem formulation. 
%For completeness, we review its problem formulation for the specific case of $N=2$ nodes.
We consider a system consisting of two nodes and a DC. There is a random variable $\mathbf{u}$, uniformly distributed in $[-M,M]$, where $M \in \mathbb{R}$. The DC aims to estimate $\mathbf{u}$ but does not have direct access to it. Instead, it relies on the two nodes to estimate $\mathbf{u}$\footnote{For simplicity of exposition, in this conference paper we focus on system with two nodes. Based on the results of~\cite{akbari2024game}, the result of this paper can be directly used for system with more than two nodes.}. One of the nodes is honest, while the other is adversarial. The adversarial node is chosen uniformly at random, and the DC does not know which node is adversarial.
The honest node sends $\mathbf{y}_h$ to the DC, where $\mathbf{y}_h = \mathbf{u} + \mathbf{n}_h$. The probability density function (PDF) of $\mathbf{n}_h$, denoted by $f_{\mathbf{n}_h}$, is symmetric. Also, $\Pr(|\mathbf{n}_h| > \Delta) = 0$ for some $\Delta \in \mathbb{R}$.

The adversary sends $\mathbf{y}_a = \mathbf{u} + \mathbf{n}_a$ to the DC, where $\mathbf{n}_a$ is an arbitrary noise independent of $\mathbf{u}$. The PDF of $\mathbf{n}_a$, denoted by $\gdot$, is chosen by the adversary. The DC is unaware of this choice. We assume that the values of $M$ and the distribution of $\mathbf{n}_h$ are known to all, and $\Delta \ll M$.

\begin{figure}[t]
    \centering
\includegraphics[width=0.98\linewidth]{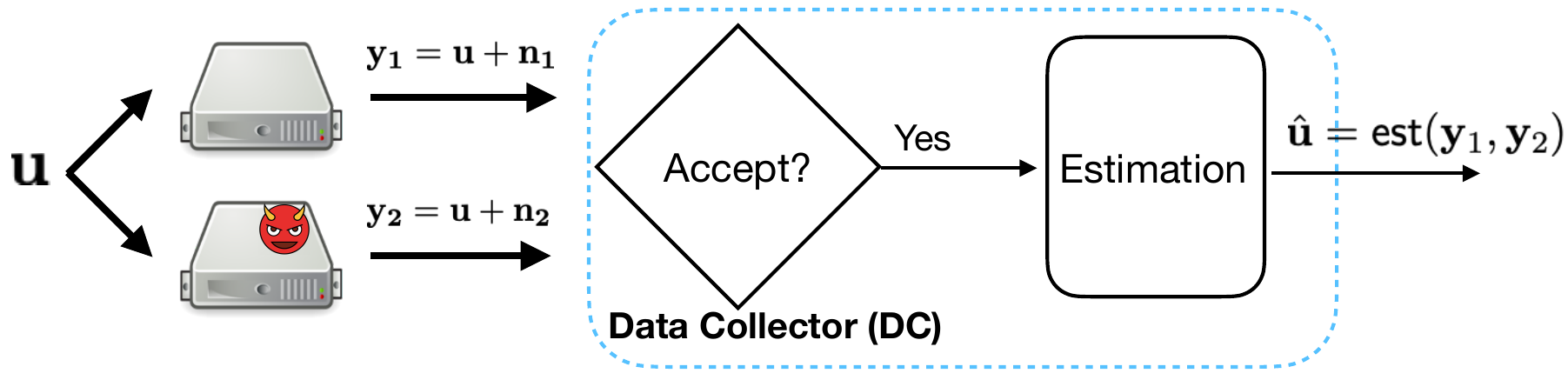}
\caption{
Game of Coding involves with one honest and one adversarial node providing noisy versions of $\mathbf{u}$
to the DC, which decides whether to accept the inputs and estimate $\mathbf{u}$. While the honest node's noise distribution is known, the adversary selects its noise to optimize its utility function in response to the DC's acceptance strategy. The DC also strategically commits to an acceptance rule, optimizing its utility function. This paper addresses practical scenarios where the DC is unaware of the adversary’s utility function.
}
\label{fig:Two_node_model}
\end{figure}

The DC receives \( \underline{\by} \triangleq \{\by_1, \by_2\} \) and evaluates the results through a two-step process (see Fig.~\ref{fig:Two_node_model}). 

\begin{enumerate}
    \item \textbf{Accept or Reject}: The DC \emph{accepts} $\underline{\by}$ if and only if $|\by_1 - \by_2| \leq \eta \Delta$ for some $\eta \geq 2$. We denote by $\acce$ the event that the inputs are accepted. Additionally, we define $\PA \left( \gdot, \pare \right) \triangleq \Pr(\acce)$.
    
    \item \textbf{Estimation}: If the inputs are accepted the DC outputs $\est(\underline{\by})$  
 as its estimate of $\bu$,  where $\est: \mathbb{R}^2 \to \mathbb{R}$. We define the cost of estimation as $\mathsf{MSE}\left(\est(.,.), \gdot, \pare\right) \triangleq \mathbb{E}\left[\big(\mathbf{u} - \est(\mathbf{\underline{y}})\big)^2 | \acce\right]$.
%which is the mean square error (MSE) for the estimation function $\est$, given that observed data has been accepted.
The DC chooses $\est$ to minimize $\MSE$. We define $ \est^*_{\pare,g} \triangleq \underset{\est: \mathbb{R}^2 \to \mathbb{R}}{\arg\min} ~\mathsf{MSE}\big(\est(.,.), \gdot, \pare\big)$,
 and also the 
 minimum mean square estimation error (MMSE) as $\mathsf{MMSE}\left(\gdot, \pare \right) \triangleq \mathsf{MSE}\left( \est^*_{\pare,g}(.,.), \gdot, \pare \right)$.

\end{enumerate}

\begin{comment}
The DC receives $\underline{\by} \triangleq \{\by_1, \dots, \by_N\}$ and evaluates the results according to the following steps (See Fig.~\ref{fig:General Model}):

\begin{enumerate}
    \item \textbf{Accept or Reject}: Define $\max(\underline{\by}) \triangleq \max \by_i$ and $\min(\underline{\by}) \triangleq \min \by_i$ for $i \in [N]$. The DC \emph{accepts} $\underline{\by}$ if and only if $\max(\underline{\by}) - \min(\underline{\by}) \leq \eta \Delta$ for some $\eta \geq 2$. We denote by $\acce$ the event that the inputs are accepted. Additionally, we define $\PA \left( \gdot, \pare \right) = \Pr(\acce).$
    
    \item \textbf{Estimation}: If the inputs are accepted, the DC outputs $\est (\underline{\by}) \triangleq \frac{\max(\underline{\by}) + \min(\underline{\by})}{2}$,
     as its estimate of $\bu$. In Section \ref{Main results} we  justify this choice for the estimation function by showing that the system is sybil resistant. The cost of estimation is defined as $\MSE(\gdot, \pare) \triangleq \mathbb{E}[(\mathbf{u} - \est (\underline{\by}))^2 |~ \acce]$,
    which is the mean squared error (MSE) of the estimation, given that the inputs have been accepted.
\end{enumerate}
\end{comment}

We model the interaction between the DC and the adversary as a two-player game, where each player seeks to maximize their respective utility function. The utility functions for the DC and the adversary are defined as $\mathsf{U}_{\mathsf{DC}}\left( \gdot, \pare \right) \triangleq \mathsf{Q}_{\mathsf{DC}} (\mathsf{MSE}, \mathsf{PA})$, and $\mathsf{U}_{\mathsf{AD}}\left( \gdot, \pare \right) \triangleq Q_{\mathsf{AD}} (\mathsf{MSE}, \mathsf{PA})$, respectively.
The function $\mathsf{Q}_{\mathsf{DC}}: \mathbb{R}^2 \to \mathbb{R}$ is non-increasing with respect to its first argument and non-decreasing with respect to its second. Conversely, $Q_{\mathsf{AD}}: \mathbb{R}^2 \to \mathbb{R}$ is strictly increasing with respect to both arguments.
We assume that the adversary knows the both utility functions, whereas the DC is unaware of the adversary's utility function.

The DC selects its strategy from the set $\Lambda_{\mathsf{DC}} = \left\{ \pare \mid \pare \geq 2 \right\}$, and the adversary selects its strategy from the set $\Lambda_{\mathsf{AD}} = \left\{ \gdot \mid \gdot: \mathbb{R} \to \mathbb{R} \text{ is a valid PDF} \right\}$.

We aim to evaluate this game under \emph{Stackelberg} model, considering the DC as the leader and the adversary as the follower. More precisely, for each action \( \pare \) to which the DC commits, the adversary’s best response is defined as $\mathcal{B}^{\pare} \triangleq \underset{\gdot \in \Lambda_{\mathsf{AD}}}{\arg\max} ~ {\mathsf{U}}_\mathsf{AD}(\gdot, \pare)$.
Each element of \( \mathcal{B}^{\pare} \) provides the adversary with the same utility, but these responses may result in different utilities for the DC. We define  $\Bar{\mathcal{B}}^{\pare} \triangleq \underset{\gdot \in \mathcal{B}^{\pare}}{\arg\min} ~ {\mathsf{U}}_\mathsf{DC}(\gdot, \pare)$, and  
\begin{align}\label{eq:etastar}
    \mathsf{U}^* = \underset{\pare \in \Lambda_{\mathsf{DC}},~ g^*(.) \in \Bar{\mathcal{B}}^{\pare}}{\sup} ~ {\mathsf{U}}_\mathsf{DC}(g^*(.), \pare).
\end{align}
For any $\eta$, let
\begin{align}
    \mathsf{U}(\eta) &\triangleq  {\mathsf{U}}_\mathsf{DC}(g^*(.), \pare), \label{def:u_function} 
\end{align}
where $g^*(.) \in \Bar{\mathcal{B}}^{\pare}$.
We model the interaction between the DC and the adversary as a multi-round process. In each round \( t \geq 1 \), the DC commits to a value \( \eta_t \) as its strategy, and the adversary adjusts its strategy in response to this commitment. 
We assume that the adversary is myopic, meaning that at each round, it selects its strategy to maximize its immediate utility based on the DC's current commitment \( \eta_t \), without considering the long-term implications of its actions.
After \( T(\frac{1}{\delta}, \frac{1}{\lambda}) \) rounds, the DC finalizes its strategy by committing to \( \hat{\eta}_T \). The objective is to design an algorithm such that, for any \( \delta, \lambda > 0 \), there exists a polynomial \( T(\frac{1}{\delta}, \frac{1}{\lambda}) \), ensuring that:
\begin{align}
    \Pr \left(  \mathsf{U}^* - \mathsf{U}(\hat{\eta}_T) > \lambda \right) < \delta.
\end{align}

\begin{comment}
\begin{definition}\label{proper_paier_def}
    We call a pair of utility functions $({\mathsf{U}}_\mathsf{DC}(.,.), {\mathsf{U}}_\mathsf{AD}(.,.))$ a \emph{proper pair} is for any $t \in [N]$, $r < t$, $g^*_t(.) \in \Bar{\mathcal{B}}^{\pare^*}$, and $g^*_r(.) \in \Bar{\mathcal{B}}^{\pare^*_{N,r}}_{N,r}$, we have ${\mathsf{U}}_\mathsf{DC}\left(g^*_t(.), \pare^* \right) \geq {\mathsf{U}}_\mathsf{DC}\left(g^*_r(.), \pare^*_{N,r} \right)$.
\end{definition}
While this definition appears natural and intuitive, there are situations where Definition \ref{proper_paier_def} may not hold (see Section~\ref{Main results} for further details). Hereafter, we assume that utility pairs are proper. In Section~\ref{Main results}, we provide an algorithm to determine whether a given pair of utilities is proper.
\end{comment}

\section{Preliminaries}\label{sce:Preliminaries}
In this section, we present the key preliminaries and insights from \cite{nodehi2024game}, where the game of coding framework was first introduced. To enable the DC to commit to a strategy that achieves a utility sufficiently close to \(  \mathsf{U}^* \), denoted by $\eta^*$, the authors in \cite{nodehi2024game} proposed a novel optimization problem, defined as follows:

\begin{align}\label{C_definition}
 c_{\eta} (\alpha) \triangleq 
    \underset{\gdot \in \Lambda_{\mathsf{AD}}}{\max} \; \underset{\mathsf{PA} \left( \gdot, \pare \right) \geq \alpha}{\mathsf{MMSE}\left(\gdot, \pare \right)}.
\end{align}
Note that  $c_{\eta}(.)$, is independent of the players' utility functions. Interestingly, in \cite{nodehi2024game}, it was shown that $\eta^*$ can be derived using the function $c_{\eta}(.)$ through a two-dimensional optimization process outlined in Algorithm \ref{Alg:finding_eta}.
 The core idea behind the correctness of Algorithm \ref{Alg:finding_eta} is that for any $g^*(.) \in \mathcal{B}^{\pare}_{\mathsf{AD}}$, with $\alpha(\eta) \triangleq \mathsf{PA} \left( g^*(.), \pare \right)$, they showed  that $\mathsf{MMSE}(g^*(.), \pare) = c_{\eta}(\alpha(\eta))$, and thus we have
 \begin{align}\label{relation_u_q}
    \mathsf{U}(\eta) = \mathsf{Q}_\DC\Big(\alpha(\eta), c_{\eta}(\alpha(\eta))\Big).
\end{align}
 
 Next, to characterize $c_{\eta}(.)$, it was demonstrated in \cite{nodehi2024game} that for any $\pare \in \Lambda_{\mathsf{DC}}$, and $0 < \alpha \leq 1$, we have

\begin{align}\label{characterization of c}
    \frac{h^*_{\eta}(\alpha)}{4\alpha} - \frac{(\eta^2+4)(\eta+2)\Delta^3}{M} 
    \leq c_{\eta}(\alpha) 
    \leq \frac{h^*_{\eta}(\alpha)}{4\alpha},
\end{align}
where $h^*_{\eta}(.)$ represents the concave envelope of the function\footnote{We assume that the cumulative distribution function of $\bn_h$ is strictly increasing, which guarantees the existence of $k_{\eta}^{-1}(.)$ within $[-\Delta, \Delta]$.} 
$ h_{\eta}(q) \triangleq \nu_{\eta}(k_{\eta}^{-1} (q))$, $0 \leq q \leq 1$,  for $\nu_{\eta}(z) \triangleq \int_{z-\eta\Delta}^{\Delta} (x+z)^2f_{\bn_h}(x)\,dx$ and $k_{\eta}(z) \triangleq \int_{z-\eta\Delta}^{\Delta} f_{\bn_h}(x)  \,dx$, $(\eta-1)\Delta \leq z \leq  (\eta+1)\Delta$. Note that the approximation in \eqref{characterization of c} becomes tight as $\frac{(\eta\Delta)^3}{M} \rightarrow 0$. Consequently, we can assume $c_{\eta}(\alpha) \approx \frac{h^*_{\eta}(\alpha)}{4\alpha}$.

Note that in Algorithm \ref{Alg:finding_eta} the DC must know the utility function of the adversary. Specifically, one of the inputs to Algorithm \ref{Alg:finding_eta} is $Q_{\mathsf{AD}}(., .)$. This raises a fundamental question: what if the DC does not have knowledge of the adversary's utility function? In this paper we explore this question and find a solution for it.

\begin{algorithm}[t]
\caption{Finding the optimal Decision Region}
\label{Alg:finding_eta}
\begin{algorithmic}[1]
\State Inputs: $Q_{\mathsf{AD}}(., .), \mathsf{Q}_{\mathsf{DC}}(., .)$, $c_{\eta}(.)$ and output: $\hat{\eta}$

\State \textbf{Step 1:} Calculate the set $\mathcal{L}_{\eta} = \underset{0 < \alpha \leq 1 }{\arg\max} ~Q_{\mathsf{AD}}(c_{\eta} (\alpha), \alpha)$
\State \textbf{Step 2:} Calculate $\hat{\eta} = \underset{\pare \in \Lambda_{\mathsf{DC}}}{\arg\max} ~ \underset{\alpha \in \mathcal{L}_{\eta}}{\min} ~ \mathsf{Q}_{\mathsf{DC}} \left(c_{\eta} (\alpha), \alpha\right)$
\end{algorithmic}
\end{algorithm}

\section{Main Results}\label{Main results}
In this section we present the main results of this paper. First, consider the following definitions.
\begin{definition}
    A function \( f:\mathcal{D} \to \mathbb{R} \) is said to be \( L \)-Lipschitz if it is continues over the domain $\mathcal{D}$, and there exists a constant \( L > 0 \) such that, for any \( x, y \in \mathcal{D} \), $|f(x) - f(y)| \leq L |x - y|$.
\end{definition}
\begin{definition}\label{def:piecewise Lipschitz}
    A function \( f:[a, b] \to \mathbb{R} \) is \( (L,d) \)-piecewise Lipschitz \cite{leobacher2022exception}, if there exists points \( a = x_1 <  \dots < x_v = b \), for some \( v \in \mathbb{N} \), such that \( f \) is \( L \)-Lipschitz on each subinterval \( (x_i, x_{i+1}) \), and $x_{i+1} - x_i \geq d$, for \( i \in [v-1] \), $d \in \mathbb{R}$.
\end{definition}

For simplicity we assume that the DC chooses $\eta \in [a,b]$, where $2 \leq a \leq b$.
We define $\bestu \triangleq  \sup ~ \mathsf{U}(\eta)$, for $\pare \in [a, b]$.
At the outset, committing to a strategy that achieves a utility sufficiently close to \( \bestu \) without knowledge of the adversary's utility function might appear to be an impossible task.
 This challenge arises because, in Algorithm \ref{Alg:finding_eta}, the DC requires the adversary's utility function. However, in this paper, we introduce Algorithm \ref{Alg:best_eta_not_knowing_Uad}, which enables the DC to converge to an optimal strategy despite this limitation.

The key insight behind this algorithm lies in a remarkable result demonstrated in \cite{nodehi2024game}: for any $g^*(.) \in \mathcal{B}^{\pare}_{\mathsf{AD}}$, with $\alpha(\eta) = \mathsf{PA} \left( g^*(.), \pare \right)$, it holds that $\mathsf{MMSE}(g^*(.), \pare) = c_{\eta}(\alpha(\eta))$. This result implies that, for any value of $\eta$ committed by the DC, if the DC can estimate the value of $\PA$ corresponding to the adversary's selected noise distribution, it can also estimate the corresponding $\mathsf{MMSE}$. 
Consequently, without direct knowledge of the adversary's utility function, the DC can estimate its own utility by leveraging the estimated values of $\PA$ and $\mathsf{MSE}$. By evaluating this procedure across different values of $\eta$, the DC can identify and commit to a sufficiently effective strategy. In the following, we present Algorithm \ref{Alg:best_eta_not_knowing_Uad}, and prove its correctness in Theorem \ref{Thm:ETC}. 

\begin{theorem}\label{Thm:ETC} Assume that the function $\mathsf{U}(.)$, defined in \eqref{def:u_function}, is \( (L,d) \)-piecewise Lipschitz function. Additionally, for any $\eta \in [a,b]$ and $0 \leq \alpha \leq 1$, the function $\mathsf{Q}_\mathsf{DC}(\alpha, c_\eta(\alpha))$ is $\ell$-Lipschitz with respect to $\alpha$. For any $\delta, \lambda > 0$, and $\bestetaalg$, the output of Algorithm \ref{Alg:best_eta_not_knowing_Uad}, we have $\Pr ( \bestu - \mathsf{U}(\bestetaalg) > \lambda ) < \delta$.
\end{theorem}

\begin{algorithm}[t]
\caption{Baseline Algorithm for Optimal Decision Region Estimation}

\label{Alg:best_eta_not_knowing_Uad}
\begin{algorithmic}[1]
\State \textbf{Inputs:} \( a, b, \delta, \lambda, \ell, L, d, c_{\eta}(.), \mathsf{Q}_{\mathsf{DC}}(.) \)
\State \textbf{Output:} \( \bestetaalg \)

\State Choose \( n > (b-a)\max \left\{ \frac{2L}{\lambda}, \frac{1}{d} \right\} \) and \( k > \frac{8\ell^2}{\lambda^2} \ln(\frac{2(n+1)}{\delta}) \), where \( n, k \in \mathbb{N} \).
\State Define \( \eta_i \triangleq a + \frac{(b-a)(i-1)}{n} \), for \( i \in [n+1] \).

\For{\( r = 1, \dots, k \)}
    \For{\( i \in [n+1] \)}
        \State Commit to \( \eta_i \) and observe the inputs.
    \EndFor
\EndFor

\State Let \( N(i) \) denote the number of accepted inputs for \( \eta_i \), for \( i \in [n+1] \).

\State Calculate \( \hat{\alpha}(\eta_i, k) = \frac{N(i)}{k} \), for \( i \in [n+1] \).

\State Compute:
\[
m = \underset{i \in [n+1]}{\arg\max}~ \hat{\mathsf{U}}(\eta_i),
\]
where \( \hat{\mathsf{U}}(\eta_i) = \mathsf{Q}_{\mathsf{DC}}(\hat{\alpha}(\eta_i, k), c_{\eta_i}(\hat{\alpha}(\eta_i, k))) \),  \( i \in [n+1] \).

\State \textbf{Output:} \( \bestetaalg = \eta_m \)
\end{algorithmic}
\end{algorithm}

\begin{remark}
    The problem formulation and Theorem \ref{Thm:ETC} in this paper are presented for the special case of $N = 2$ nodes. However, based on the results of the work in \cite{akbari2024game}, they can be readily extended to cases with $N > 2$ nodes.
\end{remark}

The runtime of Algorithm \ref{Alg:best_eta_not_knowing_Uad} is \( (n+1)k \), where \( n > (b-a)\max \left\{ \frac{2L}{\lambda}, \frac{1}{d} \right\} \), \( k > \frac{8\ell^2}{\lambda^2} \ln\left(\frac{2(n+1)}{\delta}\right) \), and \( n, k \in \mathbb{N} \). In this algorithm, the DC commits to \( n+1 \) equally spaced values of \( \eta \) in the interval \([a, b]\), defined as \( \eta_i = a + \frac{(b-a)(i-1)}{n} \), for \( i \in [n+1] \). For each \( \eta_i \), the DC performs \( k \) rounds, calculating \( N(i) \), the number of accepted inputs, and then estimates the acceptance probability \( \hat{\alpha}(\eta_i, k) = \frac{N(i)}{k} \). Using this, the DC computes the estimated utility \( \hat{\mathsf{U}}(\eta_i) = \mathsf{Q}_{\mathsf{DC}}(\hat{\alpha}(\eta_i, k), c_{\eta_i}(\hat{\alpha}(\eta_i, k))) \). Finally, the DC selects the candidate \( \eta_m \) that maximizes \( \hat{\mathsf{U}}(\eta_i) \) and outputs \( \bestetaalg = \eta_m \).

While Algorithm \ref{Alg:best_eta_not_knowing_Uad} provides a robust method for determining the optimal decision region, it can be inefficient in cases where the utility function \( \mathsf{U}(.) \) exhibits significant variations across \([a, b]\). More specifically, if some candidate values of \( \eta \) are unlikely to achieve the maximum utility, they can be dynamically eliminated during the execution of the algorithm. This approach forms the foundation of Algorithm \ref{alg:better_new_alg}.

Algorithm \ref{alg:better_new_alg} introduces an adaptive mechanism to eliminate suboptimal candidates. Initially, all \( n+1 \) candidates are included in the set \( [n+1] \). However, at each round \( r \), after committing to the remaining candidates and observing the inputs, the DC calculates the acceptance probability \( \hat{\alpha}(\eta_i, r) = \frac{N(i, r)}{r} \) for \( \eta_i \in [n+1] \setminus \mathcal{J} \), where \( \mathcal{J} \) is the set of eliminated candidates. The DC then computes the utility estimates \( \hat{\mathsf{U}}(\eta_i) \) for the remaining candidates and identifies the current best candidate \( \eta_m \). Using the confidence interval \( \epsilon_r = 2\ell \sqrt{\frac{\ln\left(\frac{4(n+1)}{\delta}\right)}{2r}} \), the DC eliminates any candidate \( \eta_i \) for which \( \hat{\mathsf{U}}(\eta_m) - \hat{\mathsf{U}}(\eta_i) > \epsilon_r \).

In the worst-case scenario, where \( \mathsf{U}(.) \) changes only slightly across \([a, b]\), the DC may not be able to eliminate many candidates, leading to a runtime similar to Algorithm \ref{Alg:best_eta_not_knowing_Uad}. However, in favorable cases, where \( \mathsf{U}(.) \) exhibits sharper variations, the DC can significantly reduce the number of candidates, thus improving runtime efficiency. 

Algorithm \ref{alg:better_new_alg} ensures that the final output \( \bestetaalgnew = \eta_m \) achieves a utility that is sufficiently close to the true optimal utility \( \bestu \) while significantly reducing computational overhead.
 In the following theorem, we establish the correctness and performance guarantees of this improved algorithm.

\begin{theorem}\label{Thm:ETC_new_better} Assume that the function $\mathsf{U}(.)$, defined in \eqref{def:u_function}, is \( (L,d) \)-piecewise Lipschitz function. Additionally, for any $\eta \in [a,b]$ and $0 \leq \alpha \leq 1$, the function $\mathsf{Q}_\mathsf{DC}(\alpha, c_\eta(\alpha))$ is $\ell$-Lipschitz with respect to $\alpha$. For any $\delta, \lambda > 0$, and $\bestetaalgnew$, the output of Algorithm \ref{alg:better_new_alg}, we have $\Pr ( \bestu - \mathsf{U}(\bestetaalgnew) > \lambda ) < \delta$.
\end{theorem}

\begin{algorithm}[t]
\caption{Enhanced Algorithm with Dynamic Candidate Elimination for Optimal Decision Region Estimation}

\label{alg:better_new_alg}
\begin{algorithmic}[1]
\State \textbf{Inputs:} \( a, b, \delta, \lambda, \ell, L, d, c_{\eta}(.), \mathsf{Q}_{\mathsf{DC}}(.) \)
\State \textbf{Output:} \( \bestetaalgnew \)

\State Choose \( n > (b-a)\max \left\{ \frac{2L}{\lambda}, \frac{1}{d} \right\} \) and \( k > \frac{8\ell^2}{\lambda^2} \ln(\frac{2(n+1)}{\delta}) \), where \( n, k \in \mathbb{N} \).
\State Define \( \eta_i \triangleq a + \frac{(b-a)(i-1)}{n} \).
\State Initialize \( \mathcal{J} \gets \emptyset \) (the set of eliminated candidates).

\For{\( r = 1, \dots, k \)}
    \For{each \( i \in [n+1] \setminus \mathcal{J} \)}
        \State Commit to \( \eta_i \) and observe the inputs.
        \State Let \( N(i, r) \) denote the number of accepted inputs for \( \eta_i \).
        \State Calculate \( \hat{\alpha}(\eta_i, r) = \frac{N(i, r)}{r} \).
        \State Compute \( \hat{\mathsf{U}}(\eta_i) = \mathsf{Q}_{\mathsf{DC}}\big(\hat{\alpha}(\eta_i, r), c_{\eta_i}(\hat{\alpha}(\eta_i, r))\big) \).
    \EndFor
    \State Find \( m = \underset{i \in [n+1] \setminus \mathcal{J}}{\arg\max} \hat{\mathsf{U}}(\eta_i) \).
    \State Let \( \epsilon_r = 2\ell \sqrt{\frac{\ln\left(\frac{4(n+1)}{\delta}\right)}{2r}} \).
    \For{each \( i \in [n+1] \setminus \mathcal{J} \)}
        \If{\( \hat{\mathsf{U}}(\eta_m) - \hat{\mathsf{U}}(\eta_i) > \epsilon_r \)}
            \State Add \( i \) to the set \( \mathcal{J} \) (eliminate \( \eta_i \)).
        \EndIf
    \EndFor
\EndFor

\State After \( k \) rounds, find \( m = \underset{i \in [n+1] \setminus \mathcal{J}}{\arg\max} \hat{\mathsf{U}}(\eta_i) \).
\State \textbf{Output:} \( \bestetaalgnew = \eta_m \).
\end{algorithmic}
\end{algorithm}

\section{Proof of Main Results}\label{Proofs of main theorems}
In this section, we present the proof of the main results.  As described in Algorithm \ref{Alg:best_eta_not_knowing_Uad}, $\eta_i = a + \frac{(b-a)(i-1)}{n}$, and \(\mathcal{S}_q = \{\eta_1, \dots, \eta_{n+1}\}\), where $n > (b-a)\max \{ \frac{2L}{\lambda}, \frac{1}{d}\}$ and \(i \in [n+1]\). Let  $\bestetaquant \triangleq \underset{\pare \in \mathcal{S}_{q}}{\arg\max} ~ \mathsf{U}(\eta)$.  Consider the following lemma.
\begin{comment}
    \begin{lemma}\label{lemma:bound_forquant}
    Let \( w:[e,t] \to \mathbb{R} \) be an \( L \)-Lipschitz function, and let \( \mathcal{J}_q \) be a finite set of numbers in \([e,t]\), such that for any \( x \in [e,t] \), there exists some \( y \in \mathcal{J}_q \) satisfying \( |x-y| \leq (t-e)z \), for some $z > 0$. Define \( x^*_{q} \triangleq \underset{x \in \mathcal{J}_q}{\arg\max} ~ w(x) \) and \( x^*_{e,t} \triangleq \underset{x \in [e,t]}{\arg\max} ~ w(x) \). We have $w(x^*_{e,t}) - w(x^*_{q}) \leq L(t-e)z$.

\end{lemma}
\end{comment}
\begin{lemma}\label{lemma:bound_forquant}
    For any $n > \frac{b-a}{d}$, we have
    $\bestu - \mathsf{U}(\bestetaquant) \leq \frac{L(b-a)}{n}.$
\end{lemma}

\begin{proof}
Since the function \( \mathsf{U}(.) \) is \( (L,d) \)-piecewise Lipschitz over \([a, b]\), based on Definition \ref{def:piecewise Lipschitz}, there exists a finite number of points \( a = x_1 < x_2 < \dots < x_v = b \), for some \( v \in \mathbb{N} \), such that \( \mathsf{U}(.) \) is \( L \)-Lipschitz on each subinterval \( (x_i, x_{i+1}) \), for \( i \in [v-1] \), and $x_{i+1} - x_i \geq d$, for \( i \in [v-1] \). Also, recall that $\bestu =  \sup ~ \mathsf{U}(\eta)$, for $\pare \in [a, b]$. Thus for any \( \epsilon > 0 \), there exists some \( \eta_{\epsilon} \in [a, b] \) such that \( \mathsf{U}(\eta_{\epsilon}) \geq \bestu - \epsilon \). Assume that we have $\eta_i \leq \eta_{\epsilon} \leq \eta_{i+1} $, for some $i \in [n]$. Additionally, assume that $x_j \leq \eta_{\epsilon} \leq x_{j+1}$, for some $i \in [v-1]$.
Note that $x_{j+1} - x_{j} \geq d > \frac{b-a}{n} = \eta_{i+1} - \eta_{i}$. This implies that we either have $x_j < \eta_i$, or $x_{j+1} > \eta_{i+1}$. Without loss of generality, assume that have $x_j < \eta_i$. Thus, $\eta_{\epsilon}, \eta_i \in (x_i, x_{i+1})$. One can verify that
\begin{align}
    \bestu &- \epsilon  \overset{(a)}{\leq} \mathsf{U} (\eta_{\epsilon})  \overset{(b)}{\leq} \mathsf{U}(\eta_i)   + L(\eta_{\epsilon} - \eta_i)\nonumber \\
    &\overset{(c)}{\leq} \mathsf{U}(\eta_i)   + \frac{L(b-a)}{n} \overset{(d)}{\leq} \mathsf{U}(\bestetaquant)   + \frac{L(b-a)}{n}
\end{align}
where (a) follows from the definition of \( \eta_{\epsilon} \), (b) follows from the fact that \( \eta_{\epsilon}, \eta_i \in (x_i, x_{i+1}) \) and \( \mathsf{U}(.) \) is \( L \)-Lipschitz on \( (x_i, x_{i+1}) \), (c) follows from \( \eta_{\epsilon} \in [\eta_i, \eta_{i+1}] \) and \( \eta_{i+1} - \eta_i = \frac{b-a}{n} \), and (d) follows from the definition of \( \bestetaquant \).
Therefore, for any \( \epsilon > 0 \), we have \( \bestu - \mathsf{U}(\bestetaquant) \leq \frac{L(b-a)}{n} + \epsilon \). Taking the limit as \( \epsilon \to 0 \), the proof of this lemma is complete.

\end{proof}

As described in Algorithm \ref{Alg:best_eta_not_knowing_Uad}, during round \(r \in [(n+1)k]\), the DC commits to \(\eta_j\), where \(j = \lceil \frac{r}{k} \rceil\). Additionally, for \(i \in [n+1]\), \(N(i)\) denotes the number of times the inputs from the nodes are accepted while the DC commits to \(\eta_i\) . For \(i \in [n+1]\), and  \(\eta_i \in \mathcal{S}_q\), we define  \(\hat{\alpha}(\eta_i, k) = \frac{N(i)}{k}\), and
\begin{align}\label{def:u_hat}    \hat{\mathsf{U}}(\eta_i) = \mathsf{Q}_{\mathsf{DC}}(\hat{\alpha}(\eta_i, k), c_{\eta_i}(\hat{\alpha}(\eta_i, k))).
\end{align}

  For any $\epsilon >0 $, we define the event of $\mA_{\epsilon} \triangleq \{ 
  \exists~ \eta \in \mathcal{S}_q:  \hat{\mathsf{U}}(\eta) > \hat{\mathsf{U}}(\bestetaquant), \nonumber \And \mathsf{U}(\bestetaquant) - \mathsf{U}(\eta) > \epsilon
 \}$. Consider the following lemma.
\begin{lemma}\label{lemma:bound_for_alg_vs_quant}
    For any $\epsilon > 0$, $\Pr (\mA_{\epsilon}) \leq 2(n+1)\exp(\frac{-k\epsilon^2}{2\ell^2} )$.
\end{lemma}
\begin{proof}

For any $\eta \in \mathcal{S}_q$, we define $\zeta_{\eta} \triangleq  \hat{\mathsf{U}}(\eta) - \mathsf{U}(\eta)$. One can verify that
\begin{align}
    \mA_{\epsilon} &= \Big\{ 
  \exists~ \eta \in \mathcal{S}_q:  \hat{\mathsf{U}}(\eta) > \hat{\mathsf{U}}(\bestetaquant), \And \mathsf{U}(\bestetaquant) - \mathsf{U}(\eta) > \epsilon
 \Big\} \nonumber \\
 & = \Big\{ 
  \exists~ \eta \in \mathcal{S}_q:  \zeta_{\eta} - \zeta_{\bestetaquant} > \mathsf{U}(\bestetaquant) - \mathsf{U}(\eta), \nonumber \\ &\And \mathsf{U}(\bestetaquant) - \mathsf{U}(\eta) > \epsilon
 \Big\} \nonumber \\
 & \subseteq \Big\{ 
  \exists~ \eta \in \mathcal{S}_q:  \zeta_{\eta} - \zeta_{\bestetaquant} > \epsilon
 \Big\} \nonumber \\
 &\subseteq \Big\{ 
  \exists~ \eta \in \mathcal{S}_q:  |\zeta_{\eta}| + |\zeta_{\bestetaquant}| > \epsilon
 \Big\}.
\end{align}
This implies that
\begin{align}\label{union_inequality}
    \Pr(\mA_{\epsilon}) &\leq \Pr (\{ 
  \exists~ \eta \in \mathcal{S}_q:  |\zeta_{\eta}| + |\zeta_{\bestetaquant}| > \epsilon
 \}) \nonumber \\
 &\leq \Pr (\{ 
  \exists~ \eta \in \mathcal{S}_q:  |\zeta_{\eta}|  > \frac{\epsilon}{2}
 \}) \nonumber \\
 &  \overset{(a)}{\leq} \sum_{\eta \in \mathcal{S}_q} \Pr \Big( 
    |\zeta_{\eta}|  > \frac{\epsilon}{2}
 \Big) ,
\end{align}
where (a) follows from union bound. Let $m \triangleq \underset{\eta \in \mathcal{S}_q}{\arg \max }~ \Pr ( 
    |\zeta_{\eta}|  > \frac{\epsilon}{2}
 )$. Based on \eqref{union_inequality}, we have
\begin{align}\label{bounding_a_epsilon}
    \Pr(\mA_{\epsilon}) &\leq  (n+1) \Pr \Big( 
    |\zeta_m| > \frac{\epsilon}{2}
 \Big) \nonumber \\
 &=(n+1)\Pr \Big(| 
    \hat{\mathsf{U}}(\eta_m) - \mathsf{U}(\eta_m)| > \frac{\epsilon}{2}| \Big) \nonumber \\
    &\overset{(a)}{=}(n+1)\Pr \Big(\Big| 
    \mathsf{Q}_{\mathsf{DC}}\big(\hat{\alpha}(\eta_m,k), c_{\eta_m}(\hat{\alpha}(\eta_m,k))\big) \nonumber \\
    &- \mathsf{Q}_{\mathsf{DC}}\big( \alpha(\eta_m), c_{\eta_m}(\alpha(\eta_m)) \big) \Big| > \frac{\epsilon}{2} \Big)
    \nonumber \\
    &\overset{(b)}{\leq} (n+1)\Pr \Big( 
    \big| \hat{\alpha}(\eta_m,k) - \alpha(\eta_m) \big|  > \frac{\epsilon}{2\ell}
    \Big) \nonumber \\
    &\overset{(c)}{\leq} 2(n+1)\exp(\frac{-k\epsilon^2}{2\ell^2} ),
\end{align}
where (a) follows from \eqref{relation_u_q} and \eqref{def:u_hat}, (b) follows from the fact that for any $\eta \in [a,b]$ and $0 \leq \alpha \leq 1$, the function $\mathsf{Q}_\mathsf{DC}(\alpha, c_\eta(\alpha))$ is an $\ell$-Lipschitz function, with respect to $\alpha$, and (c) follows from Hoeffding's inequality \cite{hoeffding1994probability}.
\end{proof}

\begin{comment}
    
For any $\epsilon > 0$, we define \han{DONT NEED THIS}
\begin{align}\label{def:the_event_of_B}
  \mB_{\epsilon} \triangleq \Big\{ 
  \hat{\mathsf{U}}(\bestetalocally) > \hat{\mathsf{U}}(\bestetaquant), \And \mathsf{U}(\bestetaquant) - \mathsf{U}(\bestetalocally) > \epsilon
 \Big\} .
\end{align}
Comparing the definitions of $\mA_{\epsilon}$ and $\mB_{\epsilon}$, one can verify that
\begin{align}\label{comapring_A_B}
    \Pr (\mB_{\epsilon}) \leq \Pr (\mA_{\epsilon}) 
\end{align}
 In addition, for any $\epsilon > 0$, under the event $\mB_{\epsilon}^c$, one of the following two cases holds:  
(i) $\hat{\mathsf{U}}(\bestetalocally) \leq \hat{\mathsf{U}}(\bestetaquant)$. On the other hand, by the definition $\bestetalocally =\underset{\pare \in \mathcal{S}_q}{\arg\max} ~ \hat{\mathsf{U}}(\eta)$. Therefore, in this case, we have $\bestetaquant = \bestetalocally $.  
(ii) $\mathsf{U}(\bestetaquant) - \mathsf{U}(\bestetalocally) \leq \epsilon$.  
Combining these two cases, it follows that for any $\epsilon > 0$, under the event $\mB_{\epsilon}^c$, the following holds:
\begin{align}\label{event_B-Complement}
    \mathsf{U}(\bestetaquant) - \mathsf{U}(\bestetalocally) \leq \epsilon.
\end{align}

\end{comment}

Now based on Lemma \ref{lemma:bound_forquant}, and \ref{lemma:bound_for_alg_vs_quant}, we prove Theorem \ref{Thm:ETC}.
For any $\epsilon> 0$, we define the event of $\mB_{\epsilon} \triangleq \{  \mathsf{U}(\bestetaquant) - \mathsf{U}(\bestetaalg) > \epsilon
 \}$.
Comparing the definitions of $\mA_{\epsilon}$ and $\mB_{\epsilon}$, one can verify that for any $\epsilon > 0$, we have
\begin{align}\label{comapring_A_B}
    \Pr (\mB_{\epsilon}) \leq \Pr (\mA_{\epsilon}) .
\end{align}
Therefore, we have
\begin{align}\label{second_bound}
    &\Pr ( \bestu - \mathsf{U}(\bestetaalg) > \lambda ) \nonumber \\
    &= \Pr \Big( \bestu - \mathsf{U}(\bestetaquant) + \mathsf{U}(\bestetaquant) - \mathsf{U}(\bestetaalg) > \lambda \Big) \nonumber \\
    &\leq \Pr \Big( \max \big(\bestu - \mathsf{U}(\bestetaquant), \mathsf{U}(\bestetaquant) - \mathsf{U}(\bestetaalg) \big) > \frac{\lambda}{2}\Big) \nonumber \\
    &\overset{(a)}{=} \Pr \Big( \max \big(\bestu - \mathsf{U}(\bestetaquant), \mathsf{U}(\bestetaquant) - \mathsf{U}(\bestetaalg) \big) > \frac{\lambda}{2} \big| \mB_{\frac{\lambda}{2}}\Big)\nonumber \\
    &\times \Pr(\mB_{\frac{\lambda}{2}}) \nonumber \\
    &+ \Pr \Big( \max \big(\bestu - \mathsf{U}(\bestetaquant), \mathsf{U}(\bestetaquant) - \mathsf{U}(\bestetaalg) \big) > \frac{\lambda}{2} \big| \mB_{\frac{\lambda}{2}}^c\Big)\nonumber \\
    & \times \Pr(\mB_{\frac{\lambda}{2}}^c) \nonumber \\
    &\overset{(b)}{\leq} \Pr(\mB_{\frac{\lambda}{2}}) 
    + \Pr \Big( \bestu - \mathsf{U}(\bestetaquant)  > \frac{\lambda}{2} \big| \mB_{\frac{\lambda}{2}}^c\Big) \nonumber \\
    &\overset{(c)}{\leq} \Pr(\mA_{\frac{\lambda}{2}}) 
    + \mathbbm{1}( \frac{L(b-a)}{n} > \frac{\lambda}{2}) \nonumber \\
    &\overset{(d)}{\leq} 2(n+1)\exp\left(\frac{-k\lambda^2}{8\ell^2}\right).
\end{align}
where (a) follows from Bayes' rule, and (b) follows from the facts that  by definition, under the event of $\mB_{\frac{\lambda}{2}}$, we have  $\mathsf{U}(\bestetaquant) - \mathsf{U}(\bestetaalg) > \frac{\lambda}{2}$, and $\Pr(\mB_{\frac{\lambda}{2}}^c) \leq 1$, (c) follows from \eqref{comapring_A_B}, Lemma \ref{lemma:bound_forquant} and the fact in Algorithm\ref{Alg:best_eta_not_knowing_Uad}, we choose $n > \frac{b-a}{d}$, 
and (d) follows from Lemma  \ref{lemma:bound_for_alg_vs_quant}, and the fact in Algorithm\ref{Alg:best_eta_not_knowing_Uad}, we choose $n > \frac{2(b-a)L}{\lambda}$.

Therefore, based on \eqref{second_bound}, if we choose  $k > \frac{8\ell^2}{\lambda^2} \ln(\frac{2(n+1)}{\delta})$, one can verify that $\Pr ( \bestu - \mathsf{U}(\bestetaalg) > \lambda ) < \delta$, which completes the proof of Theorem \ref{Thm:ETC}.

Now we prove Theorem \ref{Thm:ETC_new_better}. To do so, we first define 
\[
\mathcal{G}_r^{(i,j)} \triangleq \Big\{ \mathsf{U}(\eta_i) > \mathsf{U}(\eta_j), \And \hat{\mathsf{U}}_r(\eta_i) < \hat{\mathsf{U}}_r(\eta_j) - \epsilon_r \Big\},
\]
where \( i, j \in [n+1] \), \( \hat{\mathsf{U}}_r(\eta)  = \mathsf{Q}_{\mathsf{DC}}\big(\hat{\alpha}(\eta, r), c_{\eta}(\hat{\alpha}(\eta, r))\big)\) is the empirical utility estimate after \( r \) rounds, as described in Algorithm \ref{alg:better_new_alg}, and \( \epsilon_r = 2\ell \sqrt{\frac{\ln(4(n+1)/\delta)}{2r}} \) is the confidence interval.
The event \( \mathcal{G}_r^{(i,j)} \) represents a scenario where \( \mathsf{U}(\eta_i) > \mathsf{U}(\eta_j) \), but due to estimation errors, the empirical utility \( \hat{\mathsf{U}}_r(\eta_i) \) is incorrectly perceived as being smaller than \( \hat{\mathsf{U}}_r(\eta_j) \). Such errors could potentially lead to the erroneous elimination of \( \eta_i \) in Algorithm \ref{alg:better_new_alg}. We now prove the following lemma to bound the probability of such an event.

\begin{lemma}\label{lemma:event_A_bound}
For any  \( i, j \in [n+1] \), and any round \( r \in [k] \), we have $\Pr(\mathcal{G}_r^{(i,j)}) \leq \frac{\delta}{n+1}$.
\end{lemma}

\begin{proof}
Let \( \zeta_{\eta}^r \triangleq \hat{\mathsf{U}}_r(\eta) - \mathsf{U}(\eta) \). Then, we have
\begin{align}
    \mathcal{G}&_r^{(i,j)}\nonumber \\
    &= \Big\{ \mathsf{U}(\eta_i) > \mathsf{U}(\eta_j), \And \hat{\mathsf{U}}_r(\eta_i) < \hat{\mathsf{U}}_r(\eta_j) - \epsilon_r \Big\} \nonumber \\
    &= \Big\{ \mathsf{U}(\eta_i) > \mathsf{U}(\eta_j), \And \zeta_{\eta_j}^r - \zeta_{\eta_i}^r > \epsilon_r + \mathsf{U}(\eta_i) - \mathsf{U}(\eta_j) \Big\} \nonumber \\
    &\subseteq \Big\{ \zeta_{\eta_j}^r - \zeta_{\eta_i}^r > \epsilon_r \Big\} \nonumber \\
    &\subseteq \Big\{ |\zeta_{\eta_j}^r| + |\zeta_{\eta_i}^r| > \epsilon_r \Big\} \nonumber \\
    &\subseteq \Big\{ |\zeta_{\eta_j}^r| > \frac{\epsilon_r}{2} \Big\} \cup \Big\{ |\zeta_{\eta_i}^r| > \frac{\epsilon_r}{2} \Big\}.
\end{align}
Using the union bound, we have $
\Pr(\mathcal{G}_r^{(i,j)}) \leq \Pr( |\zeta_{\eta_i}^r| > \frac{\epsilon_r}{2} ) + \Pr( |\zeta_{\eta_j}^r| > \frac{\epsilon_r}{2} )
$. On the other hand, similar to \eqref{bounding_a_epsilon}, by applying Hoeffding's inequality, we have $\Pr( |\zeta_{\eta_i}^r| > \frac{\epsilon_r}{2} ),   \Pr(|\zeta_{\eta_j}^r| > \frac{\epsilon_r}{2}) \leq 2\exp(\frac{-r\epsilon_r^2}{2\ell^2})$. 
Since \( \epsilon_r = 2\ell \sqrt{\frac{\ln(4(n+1)/\delta)}{2r}} \), we have
$\Pr(\mathcal{G}_r^{(i,j)}) \leq 4\exp(\frac{-r\epsilon_r^2}{2\ell^2}) = \frac{\delta}{n+1}$.
This completes the proof of Lemma \ref{lemma:event_A_bound}.
\end{proof}

 Assume that during the execution of Algorithm \ref{alg:better_new_alg}, \( t \leq n \) candidates have been eliminated, i.e., \( |\mathcal{J}| = t \). Specifically, let the eliminated candidates be \( \eta_{i_1}, \eta_{i_2}, \dots, \eta_{i_t} \), which were removed during rounds \( r_1, r_2, \dots, r_t \), respectively. Furthermore, assume that in round \( r_j \), for \( j \in [t] \), we have 
$
m_j = \underset{i \in [n+1] \setminus \mathcal{J}}{\arg\max} \hat{\mathsf{U}}(\eta_i)$,
indicating that \( \eta_{m_j} \) had the maximum estimated utility at round \( r_j \). 

Define \( \mathcal{K} \triangleq \{ \eta_{i_1}, \eta_{i_2}, \dots, \eta_{i_t} \} \) as the set of eliminated candidates. Recall that \( \mathcal{S}_q = \{\eta_1, \dots, \eta_{n+1}\} \) is the set of all candidates, and \( \bestetaquant = \underset{\eta \in \mathcal{S}_{q}}{\arg\max} ~ \mathsf{U}(\eta) \).

\begin{lemma}\label{bounding_bad_elimination}
    We have $\Pr (\bestetaquant \in \mathcal{K} ) \leq \frac{t\delta}{n+1}$. 
\end{lemma}

\begin{proof}
    Note that 
    \begin{align}
       \Pr (\bestetaquant \in \mathcal{K} ) &\leq \frac{t\delta}{n+1} \leq \Pr (\underset{j \in [t]}{\cup}\mathcal{G}_{r_j}^{(i_j,m_j)}) \nonumber \\
       &\overset{(a)}{\leq} \sum_{j \in [t]} \Pr(\mathcal{G}_{r_j}^{(i_j,m_j)}) 
       \overset{(b)}{\leq} \frac{t\delta}{n+1},
    \end{align}
    where (a) follows from union bound, and (b) follows from Lemma \ref{lemma:event_A_bound}.
\end{proof}
Now we proceed to prove Theorem \ref{Thm:ETC_new_better}. Note that
\begin{align}\label{bound_I}
    &\Pr ( \bestu - \mathsf{U}(\bestetaalgnew) > \lambda ) \nonumber \\
    &\overset{(a)}{=}\Pr ( \bestu - \mathsf{U}(\bestetaalgnew) > \lambda | \bestetaquant \in \mathcal{K}  ) \Pr (\bestetaquant \in \mathcal{K} ) \nonumber \\
    &+ \Pr ( \bestu - \mathsf{U}(\bestetaalgnew) > \lambda | \bestetaquant \notin \mathcal{K}  ) \Pr (\bestetaquant \notin \mathcal{K} ) \nonumber \\
    &\overset{(b)}{\leq} \Pr (\bestetaquant \in \mathcal{K} ) + \Pr ( \bestu - \mathsf{U}(\bestetaalgnew) > \lambda | \bestetaquant \notin \mathcal{K}  ) \nonumber \\
    &\overset{(c)}{\leq} \frac{t\delta}{n+1} + \Pr ( \bestu - \mathsf{U}(\bestetaalgnew) > \lambda | \bestetaquant \notin \mathcal{K}  )
\end{align}
where (a) follows from union bound, (b) follows from the fact that $\Pr ( \bestu - \mathsf{U}(\bestetaalgnew) > \lambda | \bestetaquant \in \mathcal{K}  ), \Pr (\bestetaquant \notin \mathcal{K} ) \leq 1$, (c) follows from Lemma \ref{bounding_bad_elimination}. 

Note that, similar to \eqref{second_bound}, we can show that $\Pr ( \bestu - \mathsf{U}(\bestetaalgnew) > \lambda | \bestetaquant \notin \mathcal{K}  ) \leq 2(n+1-t)\exp\left(\frac{-k\lambda^2}{8\ell^2}\right)$. Since \( k > \frac{8\ell^2}{\lambda^2} \ln(\frac{2(n+1)}{\delta}) \), we have 
\begin{align}\label{bound_II}
    \Pr ( \bestu - \mathsf{U}(\bestetaalgnew) > \lambda | \bestetaquant \notin \mathcal{K}  ) \leq \frac{\delta(n+1-t)}{n+1}.
\end{align}
Combining \eqref{bound_I}, \eqref{bound_II} completes the proof of Theorem \ref{Thm:ETC_new_better}.

\bibliographystyle{ieeetr}

\end{document}